\documentclass[a4paper,UKenglish,cleveref, autoref]{article}

\usepackage{amsmath, amsthm, amsfonts,esint}
\usepackage{mathtools}
\usepackage{xfrac}

\usepackage[margin=1in]{geometry}

\newtheorem{theorem}{Theorem}
\numberwithin{theorem}{section}
\newtheorem{corollary}[theorem]{Corollary}
\newtheorem{lemma}[theorem]{Lemma}
\newtheorem{proposition}[theorem]{Proposition}
\newtheorem{remark}[theorem]{Remark}
\newtheorem{definition}[theorem]{Definition}

\usepackage{todonotes}

\usepackage{authblk}

\author[1]{Jakub Gajarsk\'y}
\author[2]{Maximilian Gorsky}
\author[2]{Stephan Kreutzer}
\affil[1]{University of Warsaw}
\affil[2]{Technical University Berlin}

\title{Differential games, locality and model checking for FO logic of graphs}

\date{}

\newtheorem{manualtheoreminner}{Theorem}
\newenvironment{manualtheorem}[1]{%
  \renewcommand\themanualtheoreminner{#1}%
  \manualtheoreminner
}{\endmanualtheoreminner}

\usepackage[utf8]{inputenc}
\usepackage[T1]{fontenc}

\usepackage{relsize,xspace}
 \usepackage{xcolor}
 \usepackage{mathtools}
\usepackage{microtype}
\usepackage{amsmath}
 \usepackage{amssymb}
 \usepackage{amsfonts}
\usepackage{hyperref}
\usepackage{tikz}
\usepackage{bm}
 \usepackage{enumerate}
\usepackage{tikz,wrapfig}
\usetikzlibrary{shapes,decorations}

\usetikzlibrary{decorations.markings}
\newdimen\arrowsize
\newlength{\arrowlength}
\newlength{\arrowangle}
\newlength{\arrowthickness}

\pgfarrowsdeclare{slim}{slim}
{
\arrowsize=0.7pt
\advance\arrowsize by .5\pgflinewidth
\pgfarrowsleftextend{-4\arrowsize-.5\pgflinewidth}
\pgfarrowsrightextend{.5\pgflinewidth}
}
{
\advance\arrowsize by .5\pgflinewidth
\pgfsetdash{}{0pt}
\pgfsetroundjoin
\pgfsetroundcap
\pgfpathmoveto{\pgfpointpolar{\arrowangle}{\arrowlength}}
\pgfpathlineto{\pgfpointorigin}
\pgfusepathqstroke
\pgfpathmoveto{\pgfpointpolar{\arrowangle}{\arrowlength}}
\pgfpathcurveto{\pgfpointpolar{\arrowangle*1.04}{\arrowthickness+0.5*(\arrowlength-\arrowthickness)}}%
{\pgfpointpolar{\arrowangle*1.1}{\arrowthickness+0.3*(\arrowlength-\arrowthickness)}}
{\pgfpointpolar{180}{\arrowthickness}}
\pgfpathlineto{\pgfpointorigin}
\pgfusepathqfill
\pgfpathmoveto{\pgfpointpolar{-\arrowangle}{\arrowlength}}
\pgfpathlineto{\pgfpointorigin}
\pgfusepathqstroke
\pgfpathmoveto{\pgfpointpolar{-\arrowangle}{\arrowlength}}
\pgfpathcurveto{\pgfpointpolar{-\arrowangle*1.04}{\arrowthickness+0.5*(\arrowlength-\arrowthickness)}}%
{\pgfpointpolar{-\arrowangle*1.1}{\arrowthickness+0.3*(\arrowlength-\arrowthickness)}}
{\pgfpointpolar{180}{\arrowthickness}}
\pgfpathlineto{\pgfpointorigin}
\pgfusepathqfill
}

\tikzstyle{vertex}=[circle,inner sep=1.5, outer sep=2, minimum size =5pt,semithick,fill=black, draw=black]
\tikzstyle{point}=[circle,inner sep=1,fill=black, draw=black]
 \tikzstyle{path2}=[-stealth,thin,decorate,%
 decoration={snake,amplitude=6pt,segment length=#1,%
 pre length=#1/10,post length=#1/10}]
\tikzstyle{brace}=[thin,decorate,decoration=brace]
\tikzstyle{ie}=[thin,dashed,gray]

\colorlet{fillA}{gray!50}
\colorlet{fillB}{gray!15}

\usepackage{refcount}
\usepackage{wrapfig}
\usepackage{marginnote}
\usepackage[inline]{enumitem}
\usepackage{latexsym}
\usepackage{todonotes}

\definecolor{dark-blue}{rgb}{0.05,0.25,0.85}
\hypersetup{colorlinks=true,linkcolor=dark-blue,citecolor=dark-blue}

\newcommand{\N}{\mathbb{N}}

\def\cqedsymbol{\ifmmode$\lrcorner$\else{\unskip\nobreak\hfil
\penalty50\hskip1em\null\nobreak\hfil$\lrcorner$
\parfillskip=0pt\finalhyphendemerits=0\endgraf}\fi}

\usepackage[inline]{enumitem}

\newcommand{\CCC}{\mathcal{C}}
\newcommand{\DDD}{\mathcal{D}}

\newenvironment{cenv}{\begin{list}{}{%
      \setlength{\labelwidth}{1.5em}%
      \setlength{\leftmargin}{\labelwidth}%
      \addtolength{\leftmargin}{\labelsep}%
      \setlength{\listparindent}{0em}%
      \setlength{\topsep}{10pt}%
      \setlength{\itemsep}{5pt}%
      \setlength{\parsep}{0pt}%
    }
  }{
  \end{list}
}

\newcounter{claimcounter}
\newcounter{conditioncounter}

\renewenvironment{proof}[1][]
{\setcounter{claimcounter}{0}\ifthenelse{\equal{#1}{}}{\noindent\textit{Proof.
    }}{\noindent\textit{#1. }}}%
{\hspace*{1pt}\hfill$\Box$\par\bigskip}

\newenvironment{Definition.}{\begin{definition}}{\end{definition}}
\newenvironment{Theorem.}{\begin{theorem}}{\end{theorem}}
\newenvironment{Lemma.}{\begin{lemma}}{\end{lemma}}
\newenvironment{Notation.}{\begin{notation}}{\end{notation}}
\newenvironment{Proof.}{\begin{proof}}{\end{proof}}
\newenvironment{Corollary.}{\begin{corollary}}{\end{corollary}}

\let\svthefootnote\thefootnote

\begin{document}

\maketitle

\begin{abstract}
We introduce differential games for FO logic of graphs, a variant of Ehrenfeucht-Fra\"{i}ss\'e games in which the game is played on only one graph and the moves of both players restricted. We prove that, in a certain sense, these games are strong enough to capture essential information about graphs from graph classes which are interpretable in nowhere dense graph classes. This, together with the newly introduced notion of differential locality and the fact that the restriction of possible moves by the players makes it easy to decide the winner of the game in some cases, leads to a new approach to the FO model checking problem on interpretations of nowhere dense graph classes.

\end{abstract}

\section{Introduction}
\label{sec:intro}
\let\thefootnote\relax\footnotetext{Authors' e-mail addresses: gajarsky@mimuw.edu.pl, m.gorsky@tu-berlin.de,  kreutzer@tu-berlin.de}
\let\thefootnote\svthefootnote
The first-order (FO) model checking problem asks, given a graph $G$ and a sentence $\varphi$ as input, whether $G\models \varphi$. It is known that this problem is PSPACE-complete in general~\cite{stockmeyer1974complexity, vardi1982complexity}, but one can obtain efficient parameterised algorithms on many structurally restricted classes of graphs.

There has been a long line of research studying this problem on sparse graphs and  the existence of fpt algorithms was established for graphs of bounded degree~\cite{seese96}, graphs with locally bounded treewidth~\cite{frickg01}, graphs with a locally excluded minor~\cite{dawargk07}, bounded expansion graph classes~\cite{dkt13} and nowhere dense graph classes~\cite{gks17}. The positive results on non-sparse graphs fall into two categories. The first category are formed by somewhat isolated results such as~\cite{ghkost15, gajarskyetal15, FOgeom, mapgraphs} and the recent important and general result of~\cite{twinwidth}. The second category are positive results about graph classes which can be obtained from sparse graph classes by means of interpretations~\cite{bddeg,SBE} (although~\cite{mapgraphs} can also be put into this category).

One of the reasons why the research into the FO model checking has been so successful is that Gaifman's theorem~\cite{gaifman82} -- an important result which essentially states that FO logic is local -- is particularly useful in the case of sparse graphs.
Informally, Gaifman's theorem allows us to reduce the problem of determining whether a given FO formula $\varphi$ holds on a given graph $G$ to the problem of evaluating a formula $\psi(x)$ in the $r$-neighbourhood of each vertex of $G$. In case $G$ is a graph in which each vertex has a simple neighbourhood, one can evaluate $\varphi$ on $G$ efficiently. This idea leads to efficient algorithms for evaluating FO formulas on classes of graphs of bounded degree, planar graphs, and graphs with locally bounded treewidth.

One shortcoming of using Gaifman's theorem for evaluating FO formulas is that if a graph has an (almost) universal vertex, then for $r\ge 2$ the $r$-neighbourhood of any vertex is (almost) the whole graph, and therefore evaluating formulas locally on the $r$-neighbourhoods is essentially the same as evaluating them on the whole graph. Even worse, on complements of bounded degree graphs, it holds for every vertex $v$ that almost the whole graph is in the 1-neighbourhood of $v$. In such cases, one cannot use the locality-based approach directly, but has to complement the input graph $G$ to get the graph $\Bar{G}$ first, use locality on $\bar{G}$, and then translate the results back to $G$. In many cases when dealing with non-sparse graphs, there seems to be no good way how to use Gaifman's theorem at all, and either one uses a notion of locality tailor-made to the given situation (such as in~\cite{gajarskyetal15} or~\cite{twinwidth}) or does not use locality at all (for example the dynamic programming algorithm for FO (and even MSO) logic on graph classes of bounded treewidth).

In this paper we initiate a relativised approach to FO model checking, which is aimed to work on graph classes interpretable in nowhere dense graph classes and which avoids some of the issues mentioned above. Instead of focusing on the absolute notion of neighbourhood of a vertex $v$, we focus on the symmetric difference $D(u,v)$ of neighbourhoods of two vertices $u$ and $v$. Note that for the example of complements of graphs of degree at most $d$, the set $D(u,v)$ has size at most $2d+2$ while $N(v)$ is large. Thus, rather than trying to determine whether a given formula holds in the $r$-neighbourhood of any given vertex $v$ of a graph $G$, we reduce the FO model checking to the problem of evaluating formulas $\xi(x,y)$ on the \emph{differential $r$-neighbourhood} $DN_r(u,v)$ of any given pair $u,v$ of vertices of $G$. We first provide a naive definition of this notion and will give a revised definition later: $DN_1(u,v)$ is just $D(u,v)$ and for any $i>1$, the differential $i$-neighbourhood $DN_i(u,v)$ is $DN_{i-1}$ together with the union of all $D(a,b)$, where $a,b \in DN_{i-1}(u,v)$.
The number $r$ for which we will want to consider the differential neighbourhood $DN_r(u,v)$ depends on the input sentence $\varphi$. The formula $\xi_r(x,y)$ which we then want to evaluate on $DN_r(u,v)$ says "Duplicator wins the $r$-round differential game on $G$ starting from vertices $u$ and $v$", where a \emph{differential game} is newly defined version of  Ehrenfeucht-Fra\"{i}ss\'e game which is played between two vertices $u,v$ of a graph $G$ and in which the moves of the players are guaranteed to take place in $DN_r(u,v)$.

Our contributions can be briefly summarised as follows: 
\begin{enumerate}
    \item We show that the FO model checking can be reduced to deciding whether two vertices $u,v$ of a graph $G$ have the same $q$-type, i.e. whether $u \equiv_q v$. Note that this is not trivial -- even if we have access to the $q$-equivalence relation on $V(G)$, it is not clear which equivalence class corresponds to which $q$-type.
    \item  We introduce \emph{differential games} which are aimed at distinguishing vertices of different $q$-types and which are played on differential neighbourhoods. We prove that for every $q$ there exists $r$ such that the relation $u\cong^D_r v$ defined by "Duplicator wins the $r$-round differential game between $u$ and $v$" suitably approximates $\equiv_q$ on graph classes interpretable in nowhere dense graph classes. This leads to the following theorem:
\begin{manualtheorem}{6.9}
Let $\CCC$ be a class of labelled graphs interpretable in a nowhere dense class of graphs such that we can decide the winner of the $r$-round differential game in fpt runtime with respect to the parameter $r$. Then the FO model checking problem is solvable in fpt runtime on $\CCC$.
\end{manualtheorem}
    \item We then focus on graph classes on which it is possible to decide the winner of the $r$-round differential game efficiently. We make the definition of $DN_r(u,v)$ more useful by adjusting it using colours and then we define \emph{differentially simple} graph classes. These are graph classes in which each graph can be coloured with few colours in such a way that $DN_r(u,v)\cup \{ u,v \}$ comes from a class of graphs with an efficient model checking algorithm. We then show that classes of graphs interpretable in graph classes of locally bounded treewidth are differentially simple.
\end{enumerate}

\section{Overview of our approach}
\label{sec:overview}

As mentioned in the introduction, our relativised approach to FO model checking is based on determining whether two vertices $u,v$ of $G$ \emph{differ} from each other -- i.e.\ whether $u \not\equiv_q v$, which is the case whenever there is a formula $\psi(x)$ of quantifier rank $q$ such that $G\models \psi(u)$ but $G \not\models \psi(v)$. In Section~\ref{sec:mc} we show that if we can solve this problem efficiently, then we can solve the FO model checking problem efficiently as well. To be more precise, we show that if we can efficiently compute a relation $\sim_q$ on $V(G)$ such that the transitive closure of $\sim_q$ refines $\equiv_q$ and does not have too many classes, then we can construct an \emph{evaluation tree} of  size bounded in terms of $q$. This evaluation tree then allows us to determine whether $G \models \varphi$ for every sentence $\varphi$ in prenex normal form with $q$ quantifiers.
Our approach to efficiently determining whether two vertices differ, relies on the intuition that if they differ, then they will have to differ in their differential $r$-neighbourhoods.

To capture this intuition formally, we introduce \emph{semi-differential games} in Section~\ref{sec:semidiffgame}, which are a variant of the well-known Ehrenfeucht-Fra\"{i}ss\'e (EF) games. The semi-differential game is played in one graph only and if we are given two vertices $a_0,b_0$ as a starting position, we play the game with the Spoiler's moves being restricted in the following way: In his first move the Spoiler plays a vertex in $D(a_0,b_0)$ and declares the picked vertex to be either $a_1$ or $b_1$. The Duplicator picks her reply anywhere in $G$ and the picked vertex becomes $b_1$ or $a_1$ -- the `opposite' of the Spoiler's choice. More generally, in the $i$-th move, after the vertices $a_1,\ldots,a_{i-1}$ and $b_1,\ldots,b_{i-1}$ have been played, the Spoiler picks $j \in {0,\ldots, i-1}$ and a vertex in $D(a_j,b_j)$ and calls it $a_i$ or $b_i$. Again, the Duplicator replies by picking a vertex anywhere in $G$. The winner of the game is decided as in the usual EF game, by comparing the graphs induced by $(a_0, a_1, \ldots, a_m)$ and $(b_0, b_1, \ldots, b_m)$.
 
We show that there exists a function $l:\mathbb{N} \to \mathbb{N}$ such that if the Spoiler wins the standard $m$-round Ehrenfeucht-Fra\"{i}ss\'e game on a graph $G$ starting from $a_0$ and $b_0$, then he wins the $l(m)$-round semi-differential game starting from the same position. While semi-differential games have nice properties, the fact that the Duplicator's moves are not restricted in any way makes it hard to use them algorithmically.

With this in mind we introduce \emph{differential games} in Section~\ref{sec:diffgames} in which the Duplicator's moves are restricted to $D(a_j, b_j)$ as well. The connection to the notion of differential neighbourhoods now becomes clear -- the entire $r$-round differential game between two vertices $u$ and $v$ will take place in $DN_r(u,v)$. This observation seems to indicate that differential games might be more useful algorithmically.
We now briefly focus on the properties of differential games, to demonstrate that this intuition is correct (for certain graph classes). For every graph $G$ and every $r$ we can define the relation on $V(G)$ by setting $u \cong_r^D v$ if and only if Duplicator wins the $r$-round differential game starting from $u$ and $v$. Let us also denote by $u \equiv_q v$ whenever it is true that  for each formula $\psi(x)$ of  quantifier rank $q$ it holds that $G\models \psi(u)$ iff $G \models \psi(v)$. 
As mentioned before, to obtain an efficient FO model checking algorithm it is enough to be able to decide whether $u \equiv_q v$ efficiently. 
Our approach is based on the fact that for every $q$ there exists $r$ such that the transitive closure of $u \cong_r^D v$ is a refinement of $\cong_q$. (The need to use the transitive closure stems from the fact that unlike $\equiv_q$ the relation $\cong_r^D$ is not transitive.) Moreover, for any graph class $\CCC$ interpretable in a nowhere dense graph class the number of classes in the transitive closure of $\cong_r^D$ on any $G \in \CCC$ is bounded by a number depending only on $\CCC$ and $r$. 
 
 The above considerations tell us that all we need to do to obtain an efficient FO model checking algorithm on any class $\CCC$ of graphs interpretable in a nowhere dense graph class is to guarantee that we can decide the winner of the differential game efficiently on $G[DN_r[u,v]]$ for any $G \in \mathcal{C}$ and any $u,v \in V(G)$, where $G[DN_r[u,v]]$ is the \emph{closed} differential $r$-neighbourhood of $u$ and $v$. If $G[DN_r[u,v]]$ comes from a class of graphs with an efficient model checking algorithm, then it suffices to evaluate the formula $\xi_r(x,y)$, mentioned above, on $G[DN_r[u,v]]$. Thus it is enough to focus on classes of graphs such that for every $G$ and every $u,v$ their closed differential $r$-neighbourhood comes from a class of graphs with efficient FO model checking.
 It is easy to see that classes of graphs of locally bounded treewidth have this property (for any $u,v$ it holds that $DN_r[u,v] \subseteq (N_r[u] \cup N_r[v])$) and the same is true for any class of graphs $\CCC = \{\bar{G}~|~G \in \DDD\}$, where $\DDD$ is a class of graphs of locally bounded treewidth, and $\bar{G}$ denotes the complement of $G$.
 
 Aside from these simple examples, we want to be able to use our approach on richer classes of graphs, in particular on those interpretable in graph classes of locally bounded treewidth. However, it is easy to construct an example of such graph class in which there are graphs with arbitrary large treewidth (and even clique-width) and which contain many pairs of vertices $u,v$ with $DN_1[u,v] = D[u,v] = V(G)$. Fortunately, this can still be salvaged by colouring such graphs appropriately and extending the definition of differential $r$-neighbourhoods slightly, as we will show in Section~\ref{sec:diffsimple}. The idea is to colour the vertices of a graph with a bounded number of colours in such a way that, if $D(u,v)$ is too complicated, then $u$ and $v$ get different colours. This means that if the Duplicator replied to the Spoiler's move $u$ by playing $v$ then the game would be already lost for her at this point and $D(u,v)$ is irrelevant. Thus, the more useful definition of $DN_r(u,v)$ on coloured graph is as follows: $DN_1(u,v)$ is defined only for vertices of the same colour and is equal to $D(u,v)$. For $i >1$, $DN_i(u,v)$ is also defined only for vertices of the same colour and is equal $DN_{i-1}(u,v)$ together with the union of all $D(a,b)$ with $a,b \in DN_{i-1}(u,v)$ such that \emph{$a$ and $b$ have the same colour}. Again, one can see that for a graph with colours (modelled as unary relations) the differential $r$-neighbourhood defined this way contains all vertices relevant for deciding the winner of the $r$-round differential game on $G$.
 
 To illustrate that this idea has meaningful applications, we prove that for any class $\CCC$ of graphs interpretable in a graph class of locally bounded treewidth there exists $m$ such that for any $G \in \CCC$ it is possible to colour the vertices of $G$ by at most $m$ colours in such a way that $G[DN_r(u,v)]$ has bounded clique-width. However, we were unable to find a polynomial-time algorithm which would compute such colourings. This is similar to the results of~\cite{SBE} and~\cite{twinwidth}, in which the existence of a model checking algorithm is proven, provided that a suitable decomposition of the input graph is given.
 

\section{Preliminaries}
\label{sec:prelim}
We use standard notation from graph theory. All graphs in this paper are finite, undirected, simple, and without loops. The depth of a rooted tree $T$ is the largest number of edges on any leaf-to-root path in $T$ and we say that a node $p$ is at depth $i$ in $T$ if the  distance of $p$ from the root of $T$ is $i$.

By $A \Delta B$ we denote the symmetric difference of two sets $A$ and $B$ defined by $A \Delta B = (A \setminus B) \cup (B \setminus A)$.

\subsection{Logic}
\label{subsec:logic}
We assume familiarity with FO logic.
We refer to~\cite{EbbingFlum} or any standard logic textbook for precise definitions. Since in the paper we only work with finite, simple, undirected graphs, to simplify the exposition we define the notions from logic and model theory for the vocabulary $\sigma = \{E, \{L_a\}_{a \in Lab}\}$ of labelled graphs. Here $E$ is a binary relation symbol, $Lab$ is a finite set of labels and each $L_a$ is a unary predicate symbol.




We say that two graphs $G$ and $H$ are $m$-equivalent, denoted by $G \equiv_m H$, if they satisfy the same FO sentences of quantifier rank $m$. For every $m$ the relation $\equiv_m$ is an equivalence with finitely many classes.

The FO $q$-type of a tuple of vertices $\Bar{a} = (a_1,\ldots, a_k) \in V(G)^k$, for a given (labelled) graph $G$, is defined as the set of formulas $\mathrm{tp}_q^G(\Bar{a}) \coloneqq \{ \psi(x_1,\ldots,x_k) \in \mathrm{FO}[\sigma] \ | \ G \models \psi(a_1,\ldots,a_k) \text{ and } \psi \text{ has quantifier rank } q \}$, where $\sigma = \{ E \}$, or $\sigma = \{E, \{L_a\}_{a \in Lab}\}$, if $G$ is labelled with elements of $Lab$.

Using the notion of $q$-types, we can more generally define for the tuples $\Bar{v} \coloneqq (v_1, \ldots, v_k)$ and $\Bar{u} \coloneqq (u_1, \ldots, u_k)$, consisting of vertices from $G$, and respectively from $H$, that $(G, \Bar{v}) \equiv^k_q (H, \Bar{u})$ if and only if $\mathrm{tp}_q^G(\Bar{v}) = \mathrm{tp}_q^H(\Bar{u})$. We will mostly be interested in the case when $G = H$; whenever we write $\Bar{v} \equiv^k_q \Bar{u}$, it is understood that $\Bar{v}$ and $\Bar{u}$ come from the same graph $G$ which is clear from the context and should we want to refer to the relation itself and need to note the graph it is based upon, we will add the graph as an index, as in $\equiv_q^{k,G}$.
Note that there exist only a finite number of formulas with a given quantifier rank and number of free variables. Therefore there also only exist a finite number of $q$-types for any given number of free variables. Thus the graph of the relation $\equiv_q^{k,G}$ has a number of components (cliques) bounded by a number depending only on $q$ and $k$.

For a graph $G$ and a tuple $\bar{v} = (v_1,\ldots, v_k)$ of vertices of $G$, we define the relation $\equiv_q^{\bar{v}}$ on $V(G)$ by setting $u \equiv_q^{\bar{v}} w$ if and only if $(v_1,\ldots, v_k, u) \equiv_q (v_1,\ldots, v_k, w)$.


\subsection{Games}
\label{subsec:games}
 Let $G$ and $H$ be two graphs, and $m \in \mathbb{N}$. The $m$-round \emph{Ehrenfeucht-Fra\"{i}ss\'e game}~\cite{fraisse1950nouvelle, fraisse1955quelques, ehrenfeucht1961application} (or EF game for short), denoted by
$\mathcal{G}_m(G,H)$, is played by two players called 
the \emph{Spoiler} and the \emph{Duplicator}.
Each player has to make $m$ moves in the course of play, the
players take turns and the Spoiler goes first in each round. In his $i$-th move the Spoiler first selects a graph, $G$ or $H$, and a vertex in this graph. If the Spoiler chooses $v_i$ in $G$ then
the Duplicator in her $i$-th move must choose an element $u_i$ in $H$. If the Spoiler chooses $u_i$ in $H$ then Duplicator in her $i$-th move must choose an element $v_i$ in $G$. The Duplicator wins if $\iota(v_i)=u_i$ is a label preserving isomorphism from $G[\{v_1, \ldots, v_m\}]$
to $H[\{u_1, \ldots, u_m\}]$. Otherwise the Spoiler wins. We say that a player has a \emph{winning strategy},
or in short that he \emph{wins} $\mathcal{G}_m(G,H)$, if it is possible for him to win
each play whatever choices are made by his opponent. We denote the fact that the Duplicator wins the $m$-round EF game between graphs $G$ and $H$ by $G \cong_m H$. The relation $\cong_m$ is an equivalence with finitely many classes for every $m$.
EF games and $m$-equivalence are connected by the following theorem.

\begin{theorem}[Corollary 2.2.9 in~\cite{EbbingFlum}]
\label{thm:equiv}
Let $G$ and $H$ be graphs and $m \in \mathbb{N}$. Then
$G \equiv_m H$ if and only if $G \cong_m H$.
\end{theorem}

A \emph{position} in $\mathcal{G}_m(G,H)$ is $((v_1, \ldots, v_k),(u_1, \ldots, u_{k'}))$, where each $v_i$ is from  $V(G)$, each $u_i$ is from  $V(H)$, and it holds that $k,k' \le m$ and $|k-k'| \le 1$. If $|k-k'|=0$ then it is the Spoiler's move, otherwise it is the Duplicator's move.

Let $G$ and $H$ be  graphs, $(v_1, \ldots, v_k)$ a tuple of vertices of $G$ and $(u_1, \ldots, u_k)$ a tuple of vertices of $H$. For every $m$ we can play the $m$-round EF game between $(v_1, \ldots, v_k)$ and $(u_1, \ldots, u_k)$, denoted as $\mathcal{G}_m((G, v_1, \ldots, v_k),(H, u_1, \ldots, u_k))$, by considering the $(k+m)$-round EF game between $G$ and $H$ in which the position $((v_1, \ldots, v_k),(u_1, \ldots, u_k))$ has been reached and starting the play from this position. If the Duplicator wins the $m$-round game between $(v_1, \ldots, v_k)$ and $(u_1, \ldots, u_k)$, we denote this by $(G, v_1, \ldots, v_k) \cong^k_m (H, u_1, \ldots, u_k)$. The following more general version of Theorem~\ref{thm:equiv} connects relations $\equiv^k_m$ and $\cong^k_m$.

\begin{theorem}[Theorem 2.2.8 in~\cite{EbbingFlum}]
\label{thm:games_types}
Let $G$ and $H$ be  graphs, $(v_1, \ldots, v_k)$ a tuple of vertices of $G$, $(u_1, \ldots, u_k)$ a tuple of vertices of $H$, and $m$ a non-negative integer.
Then $(G, v_1, \ldots, v_k) \equiv^k_m (H, u_1, \ldots, u_k)$ if and only if $(G, v_1, \ldots, v_k) \cong^k_m (H, u_1, \ldots, u_k)$.
\end{theorem}

Again we will be mostly interested in the case when $G=H$; whenever we write $(v_1, \ldots, v_k) \cong^k_m (u_1, \ldots, u_k)$ it is understood that $(v_1, \ldots, v_k)$ and $(u_1, \ldots, u_k)$ come from the same graph $G$ which is clear from the context.
When comparing two concrete tuples, we will write $\cong_m$ instead of $\cong_m^k$, since $k$ can be inferred from from the context and we will apply the same rationale to $\equiv_m^k$ as well.


\subsection{Interpretations}
\label{subsec:inter}
Let $\psi(x,y)$ be an FO formula with two free variables over the language
of (possibly labelled) graphs such that for any graph and any $u,v$ it holds
that $G \models \psi(u,v) \Leftrightarrow G \models \psi(v,u)$ and
$G \not\models \psi(u,u)$, i.e. the relation on $V(G)$ defined by the
formula is symmetric and irreflexive. From now on we will assume that
formulas with two free variables are symmetric and irreflexive
(which can easily be enforced). 
Given a graph $G$, the formula
$\psi(x,y)$ maps $G$ to a graph $H = I_{\psi}(G)$ defined by $V(H) = V(G)$ and
$E(H) = \{\{u,v\}~|~G \models \psi(u,v) \}$.  We then say that the graph $H$ is
\emph{interpreted} in $G$. Notice that even though the graph $G$ can be
labelled, our graph $H$ is not. 
This is to simplify our notation -- nevertheless, 
one may easily inherit labels from $G$ to $H$ if needed.

The notion of interpretation can be extended to graph classes as well. To a graph
class $\mathcal{C}$ the formula $\psi(x,y)$ assigns the graph class
$\mathcal{D} = I_{\psi}(\mathcal{C}) = \{H\>|~H=I_{\psi}(G),\, G \in
\mathcal{C}\}$.
We say that a graph class $\mathcal{D}$ is \emph{interpretable} in a graph class
$\mathcal{C}$ if there exists formula $\psi(x,y)$ such that
$\mathcal{D} \subseteq I_{\psi}(\mathcal{C})$.  
Note that
we do not require $\mathcal{D} = I_{\psi}(\mathcal{C})$, as we just want every
graph from $\mathcal{D}$ to have a preimage in $\mathcal{C}$.

\subsection{Gaifman's theorem}
\label{subsec:gaifman}
An FO formula $\phi(x_1, \ldots, x_l)$ is $r$-\emph{local}, sometimes denoted by
$\phi^{(r)}(x_1, \ldots, x_l)$, if for every graph $G$ and all $v_1, \ldots,
v_l \in V(G)$ it holds
$G \models \phi(v_1, \ldots, v_l) \Longleftrightarrow \bigcup_{1 \le i \le l}
N_r^{G}(v_i) \models \phi(v_1, \ldots, v_l)$, where $N_r^{G}(v)$ is the subgraph
of $G$ induced by $v$ and all vertices of distance at most $r$ from $v$.

\begin{theorem}[Gaifman's theorem, \cite{gaifman82}]\label{thm:Gaifman}
Every first-order formula with free variables $x_1, \ldots, x_l$ is equivalent to a
Boolean combination of the following
\begin{itemize}
\item Local formulas $\phi^{(r)}(x_1, \ldots, x_l)$ around $x_1, \ldots, x_l$,
and
\item Basic local sentences, i.e. sentences of the form
\end{itemize}
$$ 
	\exists x_1 \ldots \exists x_k	\left(\bigwedge_{1 \le i < j \le k} 
		dist(x_i,x_j) > 2r 
	\land \bigwedge_{1 \le i \le k} \phi^{(r)}(x_i) \right)
.$$
\end{theorem}

We will need the following simple corollary of Gaifman's theorem, in which we denote by $tp_q^r(v)$ the $r$-local $q$-type of $v$, i.e. the set of all $r$-local formulas $\psi(x)$ of quantifier rank $q$ such that $G \models \psi(v)$.
\begin{corollary}
\label{cor:gaifman}
  For every formula $\psi(x,y)$ there exist numbers $r$ and $q$ such that for every graph $G$ the following holds: If $u$ and $v$ are two vertices of $G$ such that the distance between them is more than $2r$, then whether $G \models \psi(u,w)$ depends only on $tp_q^r(v)$ and $tp_q^r(v)$.
\end{corollary}

\subsection{Graph classes}
\label{subsec:classes}
We assume familiarity with the notions of treewidth and of clique-width. We will need the following results about the latter concept.

\begin{theorem}[\cite{cw}]
Let $\CCC$ be a class of graphs which is interpretable in a graph class of bounded treewidth. Then $\CCC$ is of bounded clique-width.
\end{theorem}

\begin{theorem}[\cite{cw}]
\label{thm:cw_mc}
The FO model checking problem is solvabe in fpt runtime on classes of graphs of bounded clique-width.
\end{theorem}
We remark that Theorem~\ref{thm:cw_mc} assumes that clique-width decomposition of the input graph $G$ is provided together with $G$ and it is not known how to efficiently compute an optimal clique-width decomposition. However, one can approximate clique-width using the notion of rankwidth~\cite{cw_rw}, and rankwidth decompositions can be efficiently computed~\cite{rw_compute}.

\emph{Nowhere dense} graph classes were introduced by Ne\v{s}et\v{r}il and Ossona de Mendez. To define them we need the notion of an $r$-shallow minor. 
\begin{definition}[\cite{ND}]
For $r \in \mathbb{N}_0$, a graph $H$ is a \emph{shallow minor} at depth $r$ of $G$ if there exist disjoint subsets $V_1, \ldots, V_p$ of $V(G)$ such that
\begin{enumerate}
\item Each graph $G[V_i]$ has radius at most $r$, meaning that there exists $v_i \in V_i$ (a center of $V_i$) such that every vertex in $V_i$ is at the distance at most $r$ in $G[V_i]$;
\item There is a bijection $\psi : V(H) \rightarrow \{V_1, \ldots, V_p\}$ such that for every $u,v \in V(H)$, if $uv \in E(H)$ then there is an edge in $G$ with an endpoint each in $\psi(u)$ and $\psi(v)$.
\end{enumerate}
\end{definition}

The class of shallow minors of $G$ at depth $r$ (or \emph{$r$-shallow minors}) is denoted by $G \nabla r$. (Note that $G \nabla 0$ is the class of all subgraphs of $G$.)
This notation extends to graph classes: $\mathcal{C} \nabla r = \bigcup_{G\in \mathcal{C}} G \nabla r$.

Let $\omega(G)$ we denote the size of the largest complete subgraph of $G$.
For a class $\mathcal{C}$ of graphs we denote by $\omega(\mathcal{C)}$ the
$\mathrm{max}\{ \omega(G)|G \in \mathcal{C}\}$ and set $\omega(\mathcal{C)}= \infty $ if the maximum does not exist.

\begin{definition}[Nowhere dense \cite{ND}]
A graph class $\mathcal{C}$ is \emph{nowhere dense} if for all $r \in \mathbb{N}$ it holds that $\omega(\mathcal{C} \nabla r) < \infty$.
\end{definition}

\section{Differential model checking}
\label{sec:mc}

In this section we show that in order to efficiently decide whether $G \models \varphi$,
it is enough to efficiently solve the following problem: 
Given a labelled graph $G$, two of its vertices $u$ and $v$, and a number $q$, decide whether there exists a formula  $\psi(x)$ with quantifier rank $q$ such that $G \models \psi(u)$ and $G \not\models \psi(u)$, i.e. whether $u \equiv_q^1 v$. Moreover, it is enough to compute a relation $\sim$ (not necessarily an equivalence) such that the transitive closure of $\sim$ is a refinement of  $\equiv_q^1 v$ with the number of classes bounded in terms of $q$.

We will proceed as follows. First we introduce \emph{evaluation trees}, which are trees which encode evaluation of formulas in prenex normal form on a graph $G$ (these were used in~\cite{twinwidth} under the name `morphism tree'). Then we show that if we can compute a small set of vertices representing classes of $\equiv_q^{\bar{v}}$ efficiently (fpt w.r.t.\ $q$ and $|\bar{v}|$), then we can construct an evaluation tree of size bounded in terms of $q$ for any graph $G$ and a sentence in prenex normal form with $q$ quantifiers. An easy argument then shows that  $\equiv_q^{\bar{v}}$ over vocabulary $\sigma = \{E, \{L_a\}_{a \in Lab}\}$ is refined by $\equiv_q^1$ over vocabulary $\sigma'$, where $\sigma'$ is obtained from $\sigma$ by adding extra labels. Consequently, any relation $\sim$ such that the transitive closure of $\sim$ is a refinement of  $\equiv_q^1$ is a refinement of  $\equiv_q^{\bar{v}}$, and so to  obtain an efficient model checking algorithm one only has to be able to decide whether $u \sim v$ efficiently and guarantee that the number of equivalence classes of the transitive closure of $\sim$ is bounded in terms of $q$. In Section~\ref{sec:diffgames} we then show that for graph classes interpretable in nowhere dense graph classes we can take as $\sim$ the relation "Duplicator wins the $l(q)$-round differential game between vertices $u$ and $v$ of $G$", for some function $l$.
 
\subsection{Evaluation trees}
\label{subsec:evaltrees}
Let $G$ be a (possibly labelled) graph on $n$ vertices and $\varphi = Q_1x_1\ldots Q_q x_q \psi(x_1,\ldots,x_q)$ a sentence in prenex normal form with $q$ quantifiers, where each $Q_1$ is either $\exists$ or $\forall$ and $\psi(x_1, \ldots,x_q)$ is quantifier-free. The \emph{full evaluation $(G,\varphi)$-tree} is a labelled rooted tree $T$ of height $q$ with the following properties:
\begin{itemize}
    \item Every  non-leaf node $p$ of $T$ has exactly $n$ children, and these children are in one-to-one correspondence with $V(G)$. We denote the vertex of $G$ corresponding to any non-root node $s$ by $v(s)$.
    \item If $l$ is a leaf of $T$ and $(p_1, \ldots, p_q)$ is the tuple of nodes of $T$ which lie on the path from to the root (without the root itself) of $T$ to $l$, where $p_q = l$, then $l$ is labelled by $\top$ if $G \models \psi(v(p_{1}), \ldots, v(p_{q}))$ and $\bot$ otherwise.
    \item If $p$ is a non-leaf node at depth $i$ and $Q_{i+1}$ is $\exists$, then the label of $p$ is the set to $\top$ if at least one of its children has label $\top$ and to $\bot$ otherwise.
    \item If $p$ is an non-leaf node at depth $i$ and $Q_{i+1}$ is $\forall$, then the label of $p$ is the set to $\top$ if all its children have label $\top$ and to $\bot$ otherwise.
\end{itemize}

The full evaluation $(G,\varphi)$-tree $T$ corresponds to a brute-force evaluation of $\varphi$ on $G$, and it is easy to see that $G \models \varphi$ if and only if the root of $T$ gets label $\top$.

We will use a more general version of evaluation trees which correspond to simultaneously evaluating all formulas in prenex normal form with $q$ quantifiers on $G$. To define them we will need the notion of \emph{$(G,T)$-isomorphism}, which formalizes the following intuition.
One can think of any root-to-leaf path $r,p_1,\ldots,p_q$ in $T$ defined above as an assignment of vertices $v(p_1),\ldots, v(p_q)$ to variables $x_1, \ldots, x_q$ of a formula $\psi(x_1, \ldots, x_k)$  so that $x_i := v(p_i)$. The notion of $(G,T)$-isomorphism defined below captures the situation when two tuples $(p_1, \ldots, p_q)$ and $(s_1, \ldots, s_q)$ of nodes corresponding to root-to-leaf paths in $T$ are such that  $(v(p_1),\ldots, v(p_q))$ and $(v(s_1),\ldots, v(s_q))$ satisfy the same quantifier-free formulas with $q$ variables.

\begin{definition}
Let $T$ be a tree with root $r$ and $v$ a function from $V(T)\setminus\{r\}$ to $V(G)$. Let $r,p_1,\ldots,p_q$ and $r,s_1,\ldots,s_q$ be two root-to-leaf paths in $T$. We say that two tuples $(p_1, \ldots, p_q)$ and $(s_1, \ldots, s_q)$ are $(G,T)$-isomorphic if the function $f$ defined by $f(v(p_i)) = v(s_i)$ is an isomorphism between $G[\{v(p_1), \ldots, v(p_q)\}]$ and $G[\{v(s_1), \ldots, v(s_q)\}]$. Moreover, if $G$ is a labelled graph, then we require that $v(p_i)$ and $f(v(p_i))$ have the same labels in $G$ for each $i$. 
\end{definition}

Let $G$ be a (possibly labelled) graph on $n$ vertices and $q \in \mathbb{N}$. The \emph{full evaluation $(G,q)$-tree} is a labelled rooted tree $T$ of height $q$ with the following properties:
\begin{itemize}
    \item Every  non-leaf node $p$ of $T$ has exactly $n$ children, and these children are in one-to-one correspondence with $V(G)$. We denote the vertex of $G$ corresponding to any non-root node $s$ by $v(s)$.
    \item If $l$ is a leaf of $T$ and $(p_1, \ldots, p_q)$ where $p_q = l$ is the tuple of nodes of $T$ which lie on the path from the root of $T$ to $l$ (without the root itself), then $l$ is labelled by the $(G,T)$-isomorphism type of $p_1, \ldots, p_q$.
    \item If $p$ is an non-leaf node, then the label of $p$ is the set of labels of its children.
\end{itemize}

Informally, the full evaluation $(G,q)$-tree $T$ corresponds to a brute force evaluation of all FO sentences in prenex normal form with $q$ quantifiers. Moreover, in order to decide whether $G \models \varphi$ for any sentence $\varphi$ in prenex normal form it is enough to look at the label of the root of $T$. This can be easily proven by showing that the label of any node $p$ in the full evaluation $(G,\varphi)$-tree is determined by its label in the full evaluation $(G,q)$-tree, where $q$ is the quantifier rank of $\varphi$ and where we use the natural correspondence between the nodes of these two trees. 

Note that since for every $q$ and every vocabulary $\sigma$ of labelled graphs the number of different $(G,T)$-isomorphism types is bounded by $q$, an easy inductive argument shows that for every $q$ the number of different evaluation labels is bounded by a function of $q$.

Later we will need Proposition~\ref{prop:trees} below which can be proved easily using the observations made above. In the proposition we use for $\bar{v}=(v_1, \ldots, v_k)$ the notation $u \equiv_{m, pren}^{\bar{v}} w$ to denote that for every formula $\psi(x_1, \ldots,x_k, x_{k+1})$ in prenex normal form with $m$ quantifiers it holds $G\models \psi(v_1, \ldots, v_k, u)$ if and only if $G \models \psi(v_1, \ldots, v_k, w)$.
\begin{proposition}
\label{prop:trees}
Let $T$ be a full evaluation $(G,q)$-tree, $k < q$ and let $r, p_1, \ldots, p_k$  be the path in $T$ from the root $r$ to $p_k$. Then two children $s$ and $t$ of $p_k$ have the same evaluation label if and only if $v(s) \equiv_{m, pren}^{\bar{v}} v(t)$, where $\bar{v} = v(p_1), \ldots, v(p_k)$
\end{proposition}

\subsection{Reduced evaluation trees}
\label{subsec:redevaltrees}
While full evaluation trees tell us whether $G\models \varphi$, they are too big (of order $\mathcal{O}(n^{q})$) to be used directly for efficiently evaluating FO sentences. One can, however, obtain the same information from much smaller trees, which in some cases can be computed efficiently. 

A \emph{reduced} evaluation $(G,q)$-tree is a non-empty subtree $T'$ of the full evaluation $(G,q)$-tree $T$ such that \begin{enumerate}
    \item $T'$ inherits labels from $T$
    \item All leaves of $T'$ are at depth $q$
    \item \label{reduced} For every non-leaf node $p$ of $T'$ it holds that if $p$ has a child in label $L$ in $T$, then it also has a child with this label in $T'$.
\end{enumerate} 

In other words, we can obtain a reduced evaluation $(G,q)$-tree $T'$ from a full evaluation $(G,q)$-tree $T$ by deleting some non-root nodes together with all their descendants, but we have to respect condition~\ref{reduced} above.

Note that in particular the label of the root of $T'$ is the same as the label of the root of $T$, and so it is possible to determine whether $G \models \varphi$ from $T'$.
One can thus reduce the FO model checking problem to computing, given a graph $G$ and sentence $\varphi$ of quantifier rank $q$ in prenex normal form as input, a reduced evaluation $(T,q)$-tree $T'$  such that $T'$ has size bounded by a function of $q$.
To see that such trees of size $g(q)$ exist for some function $g$, one can consider the following simple bottom-up pruning procedure on the full evaluation $(G,q)$-tree $T$. For any node $p$ of $T$ of depth $q-1$ and for any label appearing on its children, keep exactly one child (leaf) with this label and delete the rest and mark node $p$ as reduced. For any node $p$ at depth $i < q-1$ such that all its children are reduced apply the same procedure -- keep exactly one child for any label appearing amongst the children of $p$. Since the number of possible evaluation labels is bounded for every $q$, the number of children of each node $p$ of the tree obtained from $T$ after exhaustively applying the above reduction rules it bounded by a function of $q$, and so the whole fully reduced tree has bounded size.

Let $T$ be a full evaluation $(G,q)$-tree, $T'$ a reduced evaluation $(G,q)$-tree and $F(T')$ a tree obtained from $T'$ by forgetting the evaluation labels (i.e. we keep just the structure of the rooted tree and the function $v$). Then it is easy to show that we can recover $T'$ from $F(T')$ by assigning labels to nodes of $T''$ as in the definition of full evaluation $(G,q)$-tree. This works because for every leaf $l$ of $F(T')$ the root-to-leaf path in $F(T')$ contains the same nodes as in $T'$ and $T$ (and so $l$ gets the same label in $T''$ as it got in $T$ and therefore in $T'$) and for every non-leaf node of $F(T')$ this follows from the property~\ref{reduced} above by induction.

It follows from the above considerations that it is possible to determine whether $G \models \varphi$ from any tree $T''$ such that $T'' = F(T')$ for some reduced evaluation $(G,q)$-tree $T'$. Our model checking algorithm is based on computing such a tree $T''$, the size of which is bounded by a function of $q$.

\subsection{Model checking by discerning types}
\label{subsec:mcbytypes}
\begin{lemma}
\label{lem:mc}
Let $A$ be an algorithm which takes as input a graph $G$, together with a tuple $\bar{v} = (v_1, \ldots, v_k)$ of vertices of $G$, and positive integer $p$ and computes
a set $S$ of vertices of $G$ such that
\begin{enumerate}
\item \label{blah} $S$ contains at least one vertex from each class of $\equiv^{\bar{v}}_p$, and
\item \label{blah2} $|S|$ depends only on $p$ and $k$.
\end{enumerate}
Assume that $A$ runs in time $f(p,k)\cdot |V(G)|^c$, where $c$ is a constant independent of $p$ and $k$. Then one can decide for any sentence $\varphi$ in prenex normal form with $q$ quantifiers, whether $G\models \varphi$ in time $g(q)\cdot |V(G)|^{c}$ for some function $g$.
\end{lemma}

\begin{proof}
Let $\varphi$ be a sentence in prenex normal form and let $q$ denote the number of quantifiers in $\varphi$. We will construct a tree $T$ such that $T=F(T')$ for some reduced evaluation $(G,q)$-tree $T'$ and such that the branching of $T$ is bounded in terms of $p$. This implies that $T$ is of bounded size (in terms of $q$) and since we can determine whether $G \models \varphi$ from $T$, the result will follow.

We first run the algorithm $A$ on $G$ with the empty tuple of vertices (i.e. k=0) and $p:=
q-1$; let $S$ be the output. For each vertex $w$ in $S$ we create a new son $t$ of the root of $T$ and set $v(t):= w$.

To form the set of children of any node $s$ of $T$ which is at depth $i$, we do the following. Let $v_1, \ldots, v_i = v$ be the vertices on the root to leaf path from the root of $T$ to $v$. We run $A$ on $(v_1, \ldots, v_i)$ and $p:=q-i$ to find $S$, and again we create a new son $t$ of $s$ for every vertex $w \in S$ and set $v(t):=w$.

The fact that tree $T$ constructed this way is a reduced tree follows from assumption~\ref{blah} about $A$, the fact that $\equiv_{q-k}^{\bar{v}}$ is a refinement of $\equiv^{\bar{v}}_{q-k, pren}$ and Proposition~\ref{prop:trees} which establishes the correspondence between classes of $\equiv^{\bar{v}}_{q-k, pren}$ and labels in a full evaluation tree.

The branching at each node of $T$ is bounded by $|S|$, which by assumption~\ref{blah2} on $A$ is bounded in terms of $q$ and $k$, and since $k$ depends on $q$ in every call of $A$, the branching of $T$ is bounded by a function of $q$. It follows that the size of $T$ is bounded in terms of $q$ and since we run algorithm $A$ once for every node of $T$, the total runtime can be bounded by $g(q)\cdot |V(G)|^{c}$ for some function $g$.
\end{proof}

\begin{lemma}
\label{lem:tuple_labels}
For every $G$ and every tuple $\bar{v} =(v_1, \ldots, v_k)$ of vertices of $G$ there exists a labelled graph $G'$ obtained from $G$ by labelling its vertices with $2k$ labels such that $\equiv^{1,G'}_q$ (with respect to the vocabulary extended by new labels)
is a refinement of $\equiv^{\bar{v},G}_q$. 
\end{lemma}

\begin{proof}
The equivalence class of any vertex $u$ of $G$ in  $\equiv^{\bar{v}}_q$ is determined by evaluating every formula $\psi(x_1, \ldots, x_k, x_{k+1})$ of quantifier rank at most $q$ on the tuple $(v_1, \ldots, v_k, u)$. We form the labelled graph $G'$ from $G$ by giving label $i$ to vertex $v_i$ and giving label $i'$ to all neighbours of $v_i$ in $G$. One can then transform any $\psi(x_1, \ldots, x_k, x_{k+1})$ to a formula $\psi'(x)$ such that for every $u \in V(G)$ it holds that $G \models \psi(v_1, \ldots, v_k, u)$ if and only if $G' \models \psi(u)$ as follows. We replace for any $i,j \le k$ in $\psi(x_1, \ldots, x_{k+1})$ any occurrence of $E(x_i,x_j)$ by $\top$ if $v_iv_j \in E(G)$ and by $\bot$ otherwise. Similarly we replace any occurrence of $x_i = x_j$ by $\top$ whenever $v_i = v_j$ and by $\bot$ otherwise. Finally, for any $i \le k$ we replace every occurrence of $x_i = x_{k+1}$ by $L_i(x)$ and every occurrence of $E(x_i, x_{k+1})$ by $L_{i'}(x)$.
It follows that the equivalence class of $u$ in $\equiv^{\bar{v},G}_q$ is determined by the equivalence class of $u$ in $\equiv^{1,G'}_q$.
\end{proof}

\begin{corollary}
\label{cor:mc}
Let $\CCC$ be a class of graphs such that there is a function $p$ such that for every $t$, every $G \in \CCC$ with at most $t$ labels and every $q$ there is a symmetric and reflexive relation $\sim_{q,t}$ on $V(G)$ such that
\begin{enumerate}
\item \label{refines} Every class of $\equiv^{1,G}_q$ (over the vocabulary extended with $t$ labels) is a union of connected components of the graph of $\sim_{q,t}$ (in other words transitive closure of  $\sim_{q,t}$ is a refinement of $\equiv^{1,G}_q$),
\item The maximum size of any independent set in graph of $\sim_{q,t}$ is bounded by $p(q,t)$, and
\item We can decide whether $u \sim_{q,t} v$ in time $|V(G)|^c\cdot h(q,t)$.
\end{enumerate}
Then one can perform model checking on $\CCC$ for any sentence $\varphi$ in prenex normal form, with quantifier rank $q$, in time $|V(G)|^{c+1}\cdot g(q)$.
\end{corollary}

\begin{proof}
We will prove that under our assumptions we can construct the algorithm $A$ from Lemma~\ref{lem:mc} with runtime $|V(G)|^{c+1}\cdot p(q,t)\cdot h(q,t)$, the result then follows.

The algorithm $A$ gets as input a graph $G$, tuple $\bar{v} = (v_1, \ldots, v_k)$ of vertices of $G$ and number $q$. The algorithm turns $G$ into graph $G'$ by labelling it as in the proof of Lemma~\ref{lem:tuple_labels} and then proceeds as follows. It greedily finds a maximal independent set $S$ in the graph of $\sim_{q,t}$ by initially setting $S:=\emptyset$ and then going through all vertices of $G$ and adding each $v$ to $S$ whenever it holds that $v \not\sim_{q,t} w$, for all $w$ currently in $S$. 

We now argue the correctness. Any maximal independent set in the graph of $\sim_{q,t}$ has to contain at least one vertex from each connected component, and so by property~\ref{refines} of $\sim_{q,t}$ it holds that $S$ contains at least one vertex from each class of $\equiv_q^1,G$, and since by Lemma~\ref{lem:tuple_labels} the relation $\equiv_q^1,G$ is a refinement of $\equiv_q^{\bar{v}, G}$, the set $S$ contains at least one vertex from each class of $\equiv_q^{\bar{v}, G}$ as desired. The size of $S$ is bounded by $p(q,t)$ and by the construction $t=2k$, so $|S|$ depends only on $q$ and $k$ as required.

In each iteration of creating $S$ its size is bounded by $p(q,t)$ and so we run the algorithm determining whether or not $v \sim_{q,t} w$ at most $p(q,t)$ times and therefore $p(q,t)\cdot |V(G)|$ times in total. This means that the total runtime can be upper bounded by $|V(G)|^{c+1}\cdot p(q,t)\cdot h(q,t)$. 
\end{proof}
\begin{remark}
If we replace the condition one in Corollary~\ref{cor:mc} by just requiring that the number of connected components of $\sim_{q,t}$ is bounded by $p(q,t)$, then we still get an efficient model checking algorithm, but with runtime $|V(G)|^{c+2}\cdot g(q)$. This is because in this case we can evaluate  $v \sim_{q,t} w$ for every pair $v,w$ of vertices and then pick one vertex from each connected component of the graph of $\sim_{q,t}$ into the set $S$. 
\end{remark}

\section{Semi-differential EF game}
\label{sec:semidiffgame}
Based on the results from Section~\ref{sec:mc}, to perform model-checking efficiently it is enough to be able to determine whether two vertices $u$ and $v$  of a graph $G$ are of the same $q$-type. To this end we introduce   \textit{semi-differential EF games} (this section) and \emph{differential EF games} (next section). The main differences compared to the standard EF game are that the (semi-)differential game is played only on one graph and the moves of the Spoiler (and for differential game also of the Duplicator) are restricted.

For two vertices $u,v$ of a graph $G$ we denote by $D(u,v)$ the symmetric difference of their neighbourhoods, which we will call their \emph{differential neighbourhood}, i.e.
$$ D(u,v):= N(u) \Delta N(v).$$
If the graph $G$ in which we want to take the differential neighbourhood is not clear from the context, we will add the relevant graph as an index, as in $D^G(u,v)$.

\begin{definition}
  The \textit{semi-differential EF game} $\mathcal{G}^D_m(G,a_1, \ldots, a_k, b_1, \ldots, b_k)$ with $m \in \N$ rounds is defined in the same way as the standard EF game with the following differences:
\begin{enumerate}
    \item The game is played only on one graph $G$ and vertices $a_1, \ldots, a_k, b_1, \ldots, b_k$ are all from $V(G)$.\footnote{One can also think of the game being played on two copies $G_1$ and $G_2$ of graph $G$, with $a_1, \ldots, a_k \in V(G_1)$ and $b_1, \ldots, b_k \in V(G_2)$. However, we need to be able to refer to $D(a_i, b_i)$, and this is more convenient if both $a_i$ and $b_i$ are in the same graph.}
    \item The starting position is $((a_1, \ldots, a_k), (b_1, \ldots, b_k))$
    \item The game is played for $m$ rounds, and is played as if it was a $k+m$-round EF game with $a_1, \ldots, a_k$ and $b_1, \ldots, b_k$ already played in the first $k$ moves.
    \item In the $j$-th round the Spoiler is only allowed to make a move on a vertex $v \in D(a_i, b_i)$ for some $i < j$. The Spoiler decides whether this move defines $a_j$ or $b_j$, i.e. whether the position after his move is $((a_1, \ldots, a_k,v), (b_1, \ldots, b_k))$ or $((a_1, \ldots, a_k), (b_1, \ldots, b_k,v))$. In case no such $v$ exists, the Duplicator wins.
    \item Duplicator's moves are unrestricted and her reply becomes $b_j$ (if Spoiler decided that the vertex $v$ he chose becomes $a_j$) or $a_j$ (if Spoiler decided that the vertex $v$ he chose becomes $b_j$).
\end{enumerate}
\end{definition}

If the game $\mathcal{G}^{SD}_m(G,a_1, \ldots, a_k, b_1, \ldots, b_k)$ is won by the Duplicator, we write $a_1, \ldots, a_k \cong^{k,SD}_m b_1, \ldots, b_k $ and we apply the same notational conventions to $\cong^{k,SD}_m$ as we did to $\cong^k_m$.
If a vertex is played to append the tuple $(a_1, \ldots, a_k)$, we call it an \emph{$a$-move}, otherwise we call it a \emph{$b$-move}.
We will refer to the semi-differential EF game as the \emph{semi-differential game} from this point onward.

We note, that the distinction between $a$-moves and $b$-moves allows us to consider the subgraphs induced by $a_1,\ldots,a_k$, and respectively by $b_1,\ldots,b_k$, as separate graphs, despite the rules of the game relying on the difference between two neighbourhoods of corresponding vertices in these tuples.
Due to this, it is also possible that the subgraphs induced by these tuples are not connected.
Consider for example a semi-differential game on $P_4$, with $V(P_4) = [4]$ and $E(P_4) = \{ \{i,i+1\} \ | \ i \in [3] \}$, starting on the position $((a_1 = 1),(b_1 = 4))$.
The Spoiler can now play $a_2 = 3$, since $D(a_1,b_1) = \{ 2,3 \}$, and the graph $P_4[\{a_1,a_2\}]$ is not connected.

Let $\cong^{SD}_{m,G}$ denote the relation "Duplicator wins the $m$ round differential game between $u$ and $v$ on the graph $G$".
As usual, we will drop the index $G$ if the graph is clear from the context.

Semi-differential games have a direct relation to regular EF games. The rest of the section is devoted to proving that at the cost of playing more moves we can play a semi-differential game instead of a regular EF-game to distinguish two vertices.

\begin{lemma}
\label{lem:main}
For every $m \in \N$  there exists $l = l(m) \in \N$ such that for every graph $G$ it holds that if $\Bar{a} \not \cong_m \Bar{b}$, then $\Bar{a} \not \cong^{SD}_{l(m)} \Bar{b}$.
\end{lemma}

\begin{proof}
We set $l(0)\coloneqq 0$ and $l(i+1)\coloneqq 2l(i)+1$ and prove the claim by induction on $m$. For $m=0$ there is nothing to prove. For the induction step, we assume that the claim holds for $m$ and prove it for $m+1$. Let  $\Bar{a} = (a_1, \ldots, a_k) $ and $\Bar{b} = (b_1, \ldots, b_k)$ be the starting position. Since by our assumptions Spoiler has a winning strategy, there exists $v \in V(G)$ such that for every $u \in V(G)$ Spoiler has a winning strategy in the $m$-round EF game from position $((\Bar{a}, v),  (\Bar{b},u))$. In particular, there exists a winning strategy for Spoiler if $v=u$. By our induction hypothesis, Spoiler wins the $l(m)$-round semi-differential game starting from $((\Bar{a}, v),  (\Bar{b}, v))$. We fix the Spoiler's winning strategy $S$ for this semi-differential game and apply it to the position $(\Bar{a}, \Bar{b})$ (this is possible because $D(v,v) = \emptyset$ and so this is never used in $S$). Let $((\Bar{a}, a_{k+1}, \ldots, a_{k+l(m)})$,   $(\Bar{b}, b_{k+1}, \ldots, b_{k+l(m)}))$ be the position after $l(m)$ rounds. If the subgraphs of $G$ induced by $(\Bar{a}, a_{k+1}, \ldots, a_{k+l(m)})$ and $(\Bar{a}, b_{k+1}, \ldots, b_{k+l(m)})$ are not isomorphic, then Spoiler already won. If they are isomorphic, then it has to hold that $v \in D(a_i, b_i)$, for some $i \in \{k+1, \ldots, k + l(m)\}$, as otherwise Duplicator's moves would beat Spoiler's strategy $S$ in the $m$-round semi-differential game starting from position $((\Bar{a}, v),  (\Bar{b}, v))$.

Since $v \in D(a_i, b_i)$, for some $i \in \{k+1, \ldots, k + l(m)\}$, Spoiler can play $v$ in the next round. Let $u$ be Duplicator's reply. By our assumptions, Spoiler wins the $m$-round EF game from the position $((\Bar{a}, v),  (\Bar{b}, u))$. Therefore, according to the induction hypothesis, there exists a winning strategy for the Spoiler in the semi-differential game with $l(m)$ rounds from the position $((\Bar{a}, v),  (\Bar{b}, u))$; let $S'$ be this strategy. Since we only restrict the Spoiler's moves in the semi-differential game, applying $S'$ to the position $((\Bar{a}, a_{k+1}, \ldots, a_{k+l(m)}, v)$,   $(\Bar{b}, b_{k+1}, \ldots, b_{k+l(m)}, u))$ will not change the outcome and thus the Spoiler wins.
\end{proof}

By contraposition of Lemma~\ref{lem:main}, if $\Bar{a} \cong^{SD}_{l(m)} \Bar{b}$, then $\Bar{a} \cong_m \Bar{b}$.
As a  consequence of Lemma~\ref{lem:main}, we then get the following.
\begin{lemma}
\label{lem:dup}
For every $m \in \N$ and every graph $G$, there exists $l = l(m) \in \N$ such that if $\Bar{a} \cong^{SD}_{l(m)} \Bar{b}$ and $\Bar{b} \cong^{SD}_{l(m)} \Bar{c}$, then $\Bar{a} \cong^{SD}_{m} \Bar{c}$.
\end{lemma}

\begin{proof}
By contraposition of Lemma~\ref{lem:main}, we immediately have $\Bar{a} \cong_m \Bar{b}$ and $\Bar{b} \cong_m \Bar{c}$. Because of the transitivity of $\cong_m$, we conclude that $\bar{a} \cong_m \Bar{c}$ and therefore in particular $\bar{a} \cong^{SD}_m (G, \Bar{c})$.
\end{proof}

Even though we are not able to establish that the relation $\cong_{m,G}^{SD}$ is an equivalence, we get the following.

\begin{lemma}
\label{lem:components}
  For every $m \in \N$ and every graph $G$, the graph of relation $\cong_{l(m)}^{SD}$ has number of connected components bounded by a function of $m$, and each equivalence class of $\cong_m$ is a union of  connected components of the graph of  $\cong_{l(m)}^{SD}$.
\end{lemma}
\begin{proof}
   Since $\cong_{l(m)}$ is an equivalence with bounded number of classes, the graph of $\cong_{l(m)}$ has bounded number of connected components (cliques). 
   Whenever the Duplicator wins the $l(m)$-round standard EF game between two vertices $u$ and $v$, then she also wins the $l(m)$-round semi-differential game (because she plays against a weaker opponent), and so the graph of $\cong^{SD}_{l(m)}$ is a supergraph of the graph of $\cong_{l(m)}$ and clearly has bounded number of components.
   
   To prove the second part it is enough to notice that, if there was an edge between $u$ and $v$ in the graph of  $\cong^{SD}_{l(m)}$ such that $u \not\cong_m v$, then this would be a contradiction with Lemma~\ref{lem:main} (Spoiler winning the normal $m$-round game but losing the differential game with $l(m)$ rounds).
\end{proof}

\section{Differential game}
\label{sec:diffgames}
In the semi-differential game we restrict Spoilers moves to $D(a_i,b_i)$ but the Duplicator's moves are unrestricted. For the application we have in mind, we will restrict Duplicators moves to $D(a_i,b_i)$ as well, using the same $i$ picked by the Spoiler, and call the resulting game the \emph{differential game}. For every graph $G$, we define the relation $\cong_m^{D}$ on $V(G)$ by setting $u \cong_m^{D} v$ if and only if Duplicator wins the $m$-round differential game starting from $u$ and $v$. We extend this notation to tuples of vertices in the same way as we did with the previous two game definitions.

The following simple extension of Lemma~\ref{lem:components} is crucial for our approach. 

\begin{lemma}
\label{lem:diffgames_refine}
There is a function $l: \mathbb{N} \to \mathbb{N}$ such that for every $m$ it holds that every class of $\cong_m$ is a union of connected components of $\cong_{l(m)}^D$.
\end{lemma}
\begin{proof}
  By Lemma~\ref{lem:components}, each class of $\cong_m$ is a union of connected components of $\cong_{l(m)}^{SD}$. It is easy to see that if Spoiler wins the  semi-differential game between vertices $u$ and $v$, then he also wins the differential game between $u,v$ (since in the differential game Duplicator is `weaker'). Thus, the graph of $\cong_{l(m)}^D$ is a subgraph of $\cong_{l(m)}^{SD}$ on the same vertex set, and the result follows.
\end{proof}

The problem with the relation $\cong_m^{D}$  is that the graph of this relation does not have bounded number of connected components (in terms of $m$) in general. An example of this are half-graphs\footnote{A half-graph of order $n$ is a bipartite graph on vertex set $\{v_1, \ldots, v_{2n}\}$ where the two parts are formed by even and odd numbered vertices and in which  there is an edge between $v_i$ and $v_j$ where $i$ is odd and $j$ is even if $i < j$.} or ladders, in which there are many nested neighbourhoods and for which the Spoiler wins for the $1$-round differential game between any pair of vertices on one side.

There are, however, very rich classes of graphs which exclude arbitrarily large half-graphs in a very strong sense -- such classes of graphs (and more general structures) are known in model theory as \emph{stable} classes of graphs (structures). A prominent example are nowhere dense graph classes introduced by Ne\v{s}et\v{r}il and Ossona de Mendez \cite{UQWisND, ND}. On these classes of graphs we can show that the graph of the relation $\cong_{m}^{D}$ has a bounded number of connected components. Moreover, we can also show this for graph classes interpretable in nowhere dense graph classes.

Before we proceed we will establish several lemmas which relate differential games to standard EF games and to interpretations, which will be used in the proof of Theorem~\ref{thm:interp_qw}, the main result of this section.

\begin{lemma}
\label{lem:local_game}
  Let $u$ and $v$ be two vertices of a graph $G$ such that $u \cong_m v$ and the distance between $u$ and $v$ is more than $2m$. Then the Duplicator wins the $m$ round  differential game between $u$ and $v$.
\end{lemma}
\begin{proof}
  The strategy for Duplicator in the differential game is determined by her strategy in the normal EF game, where the `translation' between the two games is given as follows. The starting position in the normal EF game is given by two copies $G_1$ and $G_2$ of $G$ such that in $G_1$ vertex $u$ has been played and in $G_2$ vertex $v$ has been played. Every time Spoiler plays an $a$-move in $\mathcal{G}^D_m(G,u,v)$, the same vertex is played in $G_1$ in $\mathcal{G}_m(G_1,u,G_2,v)$, and every time Spoiler plays a $b$-move in  $\mathcal{G}^D_m(G,u,v)$ the same vertex is played in $G_2$ in  $\mathcal{G}_m(G_1,u,G_2,v)$. The Duplicator's reply in $\mathcal{G}_m(G_1,u,G_2,v)$ in either $G_1$ or $G_2$ is then played in $\mathcal{G}^D_m(G,u,v)$ in $G$.
  
  We now argue that this strategy leads to a win for Duplicator in $\mathcal{G}^D_m(G,u,v)$.
  Clearly every move in the  differential game takes place in the $m$ neighbourhoods of $u$ and $v$ and since the distance between $u$ and $v$ is at least $2m$, the set of vertices played close to $u$ and $v$ are disjoint. Since the Duplicator followed the winning strategy for the normal EF game, the graphs induced by tuples $(v, v_1,\ldots, v_m)$ and $(u, u_1, \ldots, u_m)$ are isomorphic. For the  differential game, the final position is given by tuples $(a_1, \ldots, a_m)$ and $(b_1,\ldots, b_m)$, where each $a_i$ is either $u_i$ or $v_i$ and each $b_i$ is $\{u_i,v_i\}\setminus \{a_i\}$. To prove the lemma, we have to argue that for each $i,j$ the adjacency between $a_i$ and $a_j$ is the same as the adjacency between $b_i$ and $b_j$. We distinguish the following cases:
  \begin{enumerate}
      \item $a_i = u_i$ and $a_j = u_j$. In this case we have $b_i = v_i$ and $b_j = v_j$. Since tuples $(v, v_1,\ldots, v_m)$ and $(u, u_1, \ldots, u_m)$ are isomorphic, the result follows.
      \item $a_i = v_i$ and $a_j = v_j$. This case is analogous to the previous case, with the $u$ and $v$ exchanged.
      \item $a_i = u_i$ and $a_j = v_j$. In this case $b_i = v_i$ and $b_j = u_j$. We argue that in this situation there is no edge between $a_j$ and $a_j$ and the same holds for $b_i$ and $b_j$. Without loss of generality assume that $i < j$. Vertex $u_i$ is at distance smaller than $m$ from $u$ and vertex $v_j$ is at distance at most $m$ from $v$. Since the distance between $u$ and $v$ is more than $2m$, there cannot be an edge between $u_i$  and $v_j$ (and therefore between $a_i$ and $a_j$). The same argument holds for $v_i$ and $u_j$ (and therefore $b_i$ and $b_j$).
      \item $a_i = v_i$ and $a_j = u_j$. This case is analogous to the previous one.
  \end{enumerate}
  \vspace*{-0.65cm}
\end{proof}

The next two lemmas will benefit from a simple observation derived from Lemma~\ref{lem:main}.

\begin{lemma}
\label{lem:maindiff}
For every $m \in \N$ and every graph $G$, if $\Bar{a} \not \cong_m \Bar{b}$, then $\Bar{a} \not \cong^{D}_{l(m)} \Bar{b}$, where $l$ is the function from Lemma~\ref{lem:main}.
\end{lemma}
\begin{proof}
  Assume that this is not true and $\Bar{a} \cong^{D}_{l(m)} \Bar{b}$. Then the Duplicator has a winning strategy $S$ in $\mathcal{G}^D_{l(m)}(G,\bar{a},\bar{b})$. Due to Lemma~\ref{lem:main}, this strategy can however not work in $\mathcal{G}^{SD}_{l(m)}(G,\bar{a},\bar{b})$ and therefore there exists a winning strategy for the Spoiler in this game. Since the Spoiler is constrained by the same restrictions in differential and semi-differential games, this strategy can also be used to beat $S$.
\end{proof}

\begin{lemma}
\label{lem:diffgames_imitate}
Let $G$ be a graph and $\bar{a}:=v_1, \ldots, v_k$, $\bar{b}:=b_1,\ldots, b_k$ and $w$ vertices of $G$ such that $\bar{a}w \not\cong_m \bar{b}w$. Then Spoiler has a strategy in $\mathcal{G}^D_{l(m)+1}(G,\bar{a},\bar{b})$ such that he can play $w$ at some point or he wins $\mathcal{G}^D_{l(m)+1}(G,\bar{a},\bar{b})$, where $l$ is the function from Lemma~\ref{lem:main}.
\end{lemma}
\begin{proof}
  Assume the Duplicator has a winning strategy in the $\mathcal{G}^D_{l(m)}(G,\bar{a},\bar{b})$ and let $S$ be an arbitrary winning strategy for the Duplicator in this game.
  We note that any move played by the Spoiler in $\mathcal{G}^D_{l(m)}(G,(\bar{a},w),(\bar{b},w))$ must also be playable in $\mathcal{G}^D_{l(m)}(G,\bar{a},\bar{b})$, since clearly $D(w,w) = \emptyset$.
  Therefore, we can use the strategy $S$ to play $\mathcal{G}^D_{l(m)}(G,(\bar{a},w),(\bar{b},w))$ as the Duplicator.
  However, according to Lemma~\ref{lem:maindiff}, since $\bar{a}w \not\cong_m \bar{b}w$, we also have $\bar{a}w \not\cong^D_{l(m)} \bar{b}w$ and thus the Spoiler wins $\mathcal{G}^D_{l(m)}(G,(\bar{a},w),(\bar{b},w))$.
  Let $((\Bar{a},w,a_{k+1}, \ldots, a_{l(m)}),(\Bar{b},w,b_{k+1}, \ldots, b_{l(m)}))$ be the final position of the game played using $S$ against a winning Spoiler.
  Since $S$ wins for the Duplicator in $\mathcal{G}^D_{l(m)}(G,\bar{a},\bar{b})$, there exists an index $i \in [l(m)]$ such that $w \in D(a_i,b_i)$ and therefore, if we play the Spoiler's strategy against $S$ in $\mathcal{G}^D_{l(m)+1}(G,\bar{a},\bar{b})$, we can play $w$ in the $l(m)+1$-th round at the latest.
  Since we had chosen $S$ arbitrarily, it is therefore always possible for the Spoiler to either win, or play $w$ in this game.
\end{proof}

\begin{lemma}
\label{lem:diffgames_interps}
Let $\psi(x,y)$ be an interpretation formula of quantifier rank $q$ and let $G$ and $H$ be graphs such that $H = I_{\psi}(G)$. Let $a$ and $b$ be two vertices of $G$ such that $a \cong_{(m+1)(l(q)+1)}^D b$ in $G$, where $l$ is the function from Lemma~\ref{lem:main}. Then $a \cong_{m}^D b$ in $H$.
\end{lemma}

\begin{proof}
 Let $r \coloneqq (m+1)(l(q)+1)$. Assume for contradiction that $a \not\cong_m^D b$ in $H$, i.e. that Spoiler wins $\mathcal{G}_q(H,a,b)$. We will construct a winning strategy for Spoiler in  $\mathcal{G}_{r}(G,a,b)$ which is a contradiction with $a \cong_{r}^D b$ in $G$.

To provide a Spoiler's winning strategy in $\mathcal{G}_{r}(G,a,b)$ we will show that for every $i \in \{0, \ldots, m\}$ there is a strategy $S$ for the Spoiler in $\mathcal{G}_{r}(G,a,b)$ such that he can either win $\mathcal{G}_{r}(G,a,b)$ in at most $i(l(q)+1)$ rounds or play his $i$-th move according to his strategy in $\mathcal{G}_m(H,a,b)$. By playing according to this strategy for $i=m$ in $\mathcal{G}_{r}(G,a,b)$ the Spoiler will either win or he will be in the situation where vertices $a, a_1, \ldots, a_z$ and $b, b_1, \ldots, b_z$ have been played, such that $ z \le m(l(q)+1)$ and there are vertices $a, a_{j_1}, \ldots, a_{j_{m}}$ and $b, b_{j_1}, \ldots, b_{j_{m}}$ such that for each $i \in [m]$ it holds $a_{j_i}$ corresponds to $i$-th move of the Spoiler in $\mathcal{G}_q(H,a,b)$ and $b_{j_i}$ to the Duplicator's reply or vice versa. Since $a, a_{j_1}, \ldots, a_{j_{m}}$ and $b, b_{j_1}, \ldots, b_{j_{m}}$ were played according to the Spoiler's winning strategy in $H$, there has to be a pair of indices $j_p, j_s$ among $j_1, \ldots, j_m$ such that $a_{j_p}a_{j_s} \in E(H)$ and $b_{j_p}b_{j_s} \not\in E(H)$ (or vice versa), which means that $G\models \psi(a_{j_p},a_{j_s})$ and $G\not\models \psi(b_{j_p},b_{j_s})$ (or vice versa). In either case it holds that $a_{j_p}a_{j_s} \not\cong_q b_{j_p}b_{j_s}$, and so by Lemma~\ref{lem:maindiff} the Spoiler wins the $q$ round differential game between $a_{j_p}a_{j_s}$ and $b_{j_p}b_{j_s}$. The Spoiler can thus use the remaining $(m+1)(l(q)+1) - z \ge q$ moves in $\mathcal{G}_{r}(G,a,b)$ to win the game.
 
We prove the existence of the strategy outlined above by induction on $i$. For $i=0$ there is nothing to prove. From now on we assume that $i>0$ and that the statement holds for $i-1$.
 
Let $a, a_1, \ldots, a_z$ and $b, b_1, \ldots, b_z$  be the vertices played in the first $z \le (i-1)(l(q)+1)$ rounds of $\mathcal{G}_{r}(G,a,b)$ such that for some $j_1,\ldots, j_{i-1}$ the vertices $a, a_{j_1}, \ldots, a_{j_{i-1}}$ and $b, b_{j_1}, \ldots, b_{j_{i-1}}$ are such that for each $p \in [i-1]$ it holds that  $a_{j_p}$ or $b_{j_p}$ correspond to $p$-th move of Spoiler in $\mathcal{G}_q(H,a,b)$.

Let $w$ be the vertex played by the Spoiler in the $i$-th round of $\mathcal{G}_m(H,a,b)$ whilst in the position $((a, a_{j_1}, \ldots, a_{j_{i-1}}),(b, b_{j_1}, \ldots, b_{j_{i-1}}))$ and let $p<i$ be the index he used to play $w$, i.e. $w \in D(a_{j_p},b_{j_p})$. Because $a_{j_p}w \in E(H)$ and $b_{j_p}w \not\in E(H)$ or vice versa, it has to hold that $G \models \psi(a_{j_p},w)$ and $G \not\models \psi(b_{j_p},w)$ or vice versa. In both cases $a_{j_p}w \not\cong_q b_{j_p}w$, and so by Lemma~\ref{lem:diffgames_imitate} the Spoiler wins in $l(q)+1$ moves from between $a_{j_p}$ and $b_{j_p}$, or he can play vertex $w$ in $l(q)+1$ rounds of the differential game between $a_{j_p}$ and $b_{j_p}$, and so he can also either win or play $w$ from the position $((a, a_1,\ldots, a_z),(b, b_1, \ldots,b_z))$ in $l(q)+1$ rounds.
\end{proof}

We will use the following characterisation of nowhere dense graph classes. Informally, the following definition says that we can obtain a large $r$-scattered set in any sufficiently large set $A \subseteq V(G)$ by removing a few vertices from $G$.

\begin{definition}[Uniform quasi-wideness \cite{dawar2007finite, dawar2010homomorphism}]
  A class $\CCC$ of graphs is uniformly quasi-wide if for each  $r \in \mathbb{N}$ there is a function $N: \mathbb{N} \to \mathbb{N}$ and a constant $s \in \mathbb{N}$ such that for every $k \in \mathbb{N}$, graph $G \in \CCC$ and subset $A$ of $V(G)$ with $|A| \ge N(k)$, there is a set $S$ of size $|S| \le s$ such that in $A \setminus S$ there are at least $k$ vertices with pairwise distance more than $r$ in $G\setminus S$.
\end{definition}

\begin{theorem}[\cite{UQWisND}]
  A class $\CCC$ of graphs is uniformly quasi-wide if and only if $\CCC$ is nowhere dense.
\end{theorem}

This notion can then be used to prove that for any first order interpretation $\psi(x,y)$ and any nowhere dense class $\CCC$ of graphs the number of components of $\cong^{D}_{m}$ in any $G \in I_{\psi}(\CCC)$ is bounded. 

\begin{theorem}
\label{thm:interp_qw}
  Let $\CCC$ be a uniformly quasi-wide class of labelled graphs and let $\psi(x,y)$ be a first order interpretation formula.
  Then for each $m$ there exists $p$ such that for every $H \in I_{\psi}(\CCC)$ in the graph of $\cong^{D,H}_{m}$ the maximum size of any independent set is at most $p$.
\end{theorem}

\begin{proof}
Suppose that there exists $m$ such that for every $p$ there is a graph $H \in I_{\psi}(\CCC)$ for which there is an independent set of size more than $p$ in the graph of $\cong_m^{D, H}$ .

In what follows we assume that the graph $H$ we work with is as large as necessary and maximum size of an independent set in the graph of $\cong^{D}_{m}$ is as large as we need in our argumentation. We set $A$ to be an independent set of maximum size in the graph of $\cong^{D}_{m}$. We therefore have $|A| > p$, and since we can choose $p$ to be arbitrarily large, we can ensure that $|A|$ is as large as we want. 

Let $q$ be the quantifier rank of $\psi$ and set $d \coloneqq (m+1)(l(q)+1)$, where $l$ is the function from Lemma~\ref{lem:main}.
Since $H \in I_{\psi}(\CCC)$, there exists at least one $G \in \CCC$ such that $I_{\psi}(G) = H$.
Because $\CCC$ is uniformly quasi-wide, we know (by applying the definition to $r=2d+1$) that there exists a constant $s$ and a function $N$ such that for any number $k$ and any set $A$ of size at least $N(k)$, it is possible to remove $s$ vertices from $G$ such that there are $k$ vertices $v_1, \ldots, v_k$ at pairwise distance at least $r$ in $G \setminus S$. We create graph $G'$ from $G \setminus S$ by putting vertices from $S$ back (but without any edges) and labelling them each with a different colour from $[s]$. Additionally, we label the neighbourhood in $G$ of each vertex with colour $i$ with the label $l_i$. In $G'$ the vertices $v_1,\ldots, v_k$ remain pairwise at distance more than $2m$. We can recover $G$ from $G'$ by an interpretation $\delta(x,y)$, which is quantifier-free. By concatenating $\delta$ and $\psi$, we obtain an interpretation formula $\psi'(x,y)$, with quantifier rank $q$, such that $H = I_{\psi'}(G')$.

We choose $A$ to be large enough so that among $v_1, \ldots, v_k$ there exists a pair of distinct vertices $v_a$ and $v_b$ with $v_a \cong_d^{G'} v_b$ (note $k$ only has to be larger than the number of classes of relation $\cong_r^1$ with respect to the vocabulary of $\CCC$ enriched by $2s$ labels, and so $k$ does not depend on the particular graph $G$).
Since $v_a$ and $v_b$ lie at distance $r$ in $G'$ and $r > 2d$, we can use Lemma~\ref{lem:local_game} to conclude that $v_a \cong_d^{D,G'} v_b$ is true as well. We can now use Lemma~\ref{lem:diffgames_interps} to conclude that $v_a \cong_m^{D,H} v_b$, which contradicts our assumptions because $v_a$ and $v_b$ come from an independent set of the graph of $\cong_m^{D,H}$.
\end{proof}

Combining Lemma~\ref{lem:diffgames_refine} and Theorem~\ref{thm:interp_qw} with Corollary~\ref{cor:mc} we get the following theorem.
\begin{theorem}
\label{thm:diffgames_mc}
Let $\CCC$ be a class of labelled graphs interpretable in a nowhere dense class of graphs such that we can decide the winner of the $m$-round differential game in fpt runtime with respect to the parameter $m$. Then the FO model checking problem is solvable in fpt runtime on $\CCC$.
\end{theorem}

\begin{proof}[Proof of Theorem~\ref{thm:diffgames_mc}]
Let $\CCC$ be a class of graphs with labels from the set $\{1,\ldots, t\}$ and properties assumed in the statement of the theorem. From Lemma~\ref{lem:diffgames_refine}  it follows that for every $G\in \CCC$ it holds that the closure of $\cong_{l(q)}^{D,G}$ refines $\cong_q^{1,G}$ and by Theorem~\ref{thm:games_types} the relation $\cong_q^1$ is the same as $\equiv_q^1$. It follows that the closure of $\cong_{l(q)}^{D,G}$ refines $\equiv_q^{1,G}$. By Theorem~\ref{thm:interp_qw} we know that the maximum size of and independent set in the graph of $\equiv_{l(q)}^{D,G}$ is bounded in terms of $l(q)$. It follows that we can use $\cong_{l(q)}^{D,G}$ as $\sim_{q,t}$ in Corollary~\ref{cor:mc}.
\end{proof}

To conclude this section, we show that for every $m$ there is a formula $\xi_m(x,y)$ which expresses that Duplicator wins the $m$-round differential game between tuples $\bar{x}$ and $\bar{y}$, which will be used in Section~\ref{sec:diffsimple}. While it is not difficult to see  that such formulas exist, we give the construction for completeness.
\[ \xi_0(x_1, \ldots, x_k, y_1,\ldots, y_k) \coloneqq \bigwedge_{\substack{i,j \in [k] \\ i\not= j }} (E(x_i,x_j) \leftrightarrow E(y_i, y_j)) \land \bigwedge_{\substack{i,j \in [k] \\ i\not= j }} (x_i=x_j \leftrightarrow y_i = y_j) \land \bigwedge_{i \in [k]} \bigwedge_{a \in Lab} L_a(x_i)  \]

\[ \xi_m(x_1, \ldots, x_k, y_1,\ldots, y_k) \coloneqq \bigvee_{i \in [m]} (\forall x D(x_i,y_i,x) \rightarrow (\exists y D(x_i,y_i,x) \land \xi_{m-1}(x_1, \ldots, x_k, x, y_1,\ldots, y_k, y))),\]

\noindent where $D(x,y,z) \coloneqq (E(x,z) \land \lnot E(y,z)) \lor (\lnot E(x,z) \land E(y,z)).$

\section{Differentially simple graph classes}
\label{sec:diffsimple}
Based on the results from the previous sections, in order to evaluate FO sentences in prenex normal form on a graph class $\CCC$ interpretable in a nowhere dense graph class, it is enough to be able to  determine the winner of  differential game on pairs of vertices of a graph from $\CCC$. This can, however, be too difficult -- for example consider the case when $\CCC$ is an interpretation of planar graphs. In this case for many pairs of vertices $u,v$ of $G$ from $\CCC$ it can happen that $D(u,v)$ contains most of, or even the entire, graph, and this would just be the first round. We can sidestep this problem by giving such vertices $u$ and $v$ different labels. This essentially means that whenever Duplicator would play $u$ as a reply to $v$ (or vice versa), she would already have lost from that point on, and so $D(u,v)$ would be irrelevant. Extending these ideas to more than one round leads to the following definitions.

\begin{definition}[Differential neighbourhoods]
Let $G$ be a coloured graph  and let $c(a)$ denote the colour of a vertex $a$ of $G$.
\begin{itemize}
\item The differential 1-neighbourhood $DN_1(u,v)$ of vertices $u,v$ with $c(u)=c(v)$ is the set $D(u,v)$.
\item For $r \in \mathbb{N}$ with $r > 1$,
\[ DN_r(u,v) \coloneqq \big( \bigcup_{\substack{a,b \in DN_{r-1}(u,v) \\ c(a) = c(b) } } D(a,b) \big) \cup DN_{r-1}(u,v) . \]
\item For $r \in \mathbb{N}$, the \emph{closed} differential $r$ is defined as $DN_1[u,v] \coloneqq DN_r(u,v) \cup \{ u,v \}$ 
\end{itemize} 
\end{definition}

\begin{definition}
\label{def:diff_simple}
We say that class $\CCC$ is \emph{differentially simple} if for every $r$ there exists $m_r \in \mathbb{N}$ and a graph class $\DDD_r$ with efficient FO model checking algorithm such that it is possible to colour every $G \in \CCC$ with $r_m$ colours such that for every pair of $u,v$ of vertices of the same colour it holds that $G[DN_r[u,v]]$ is a graph from $\DDD_r$.
\end{definition}

Note that if $\CCC$ is interpretable in a nowhere dense class of graphs, then adding at most $m$ labels to each graph from $\CCC$ does not change the fact the maximum size of an independent set in the graph of $\cong_r^D$ for each $G$ from $\mathcal{C}$ is bounded (because Theorem~\ref{thm:interp_qw} works with labelled graphs). This allows us to focus on determining the winner of the differential game on $G[DN_r[u,v]]$, which is from $\DDD_r$, instead of on $G$ which is from $\CCC$. In case $\DDD_r$ is a class of graphs with efficient model checking algorithm, we can use this to determine the winner of the game.
We will use the FO formula $\xi_r(x,y)$, expressing that Duplicator wins the $r$-round differential game between $x$ and $y$, which is defined at the end of the previous section, and evaluate it on $G[DN_r[u,v]]$ using the model checking algorithm for $\DDD_r$. This is summarised in the following theorem.

\begin{theorem}
\label{thm:mcdiffsimple}
Let $\CCC$ be a differentially simple class of graphs such that
\begin{itemize}
    \item $\CCC$ is interpretable in a nowhere dense class of graphs, and
    \item There exists an fpt algorithm (with respect to the parameter $r$) which computes the colouring from Definition~\ref{def:diff_simple} for every $G \in \CCC$.
\end{itemize} 
Then the FO model checking problem is in FPT on $\CCC$.
\end{theorem}

\begin{proof}
  Given graph $G$ from $\CCC$ and sentence $\varphi$ in prenex normal form with $q$ quantifiers as input, we first compute the $m_{l(q)}$-colouring (where $l$ is the function from Lemma~\ref{lem:main}) of $G$ from Definition~\ref{def:diff_simple}.  We generate for all pairs $u,v$ of vertices of $G$ the differential $l(r)$-neighbourhood $DN_r[u,v]$ in polynomial time. Since $G[DN_{l(r)}[u,v]] \in \DDD_r$, we can evaluate the formula $\xi_{l(r)}(x,y)$ on $G[DN_{l(r)}[u,v]]$ efficiently and so we can decide the winner of the $l(r)$-round differential game played between $u$ and $v$. We can thus compute the relation $\cong_{l(r)}^D$ in fpt runtime, and by Theorem~\ref{thm:diffgames_mc} the FO model checking is in FPT.
\end{proof} 

To show that differentially simple graph classes can be useful, we prove that classes of graphs interpretable in graph classes with locally bounded treewidth are differentially simple, where each graph class $\DDD_r$ is a class of graphs of bounded clique-width. We note that in this case there is one $m$-colouring which works for every value of $r$ and which satisfies the requirements of Definition~\ref{def:diff_simple}. 

\begin{lemma}
\label{lem:lbdtw_diffsimple}
Let $\CCC$ be a class of graphs which is interpretable in a class of graphs of locally bounded treewidth. Then $\CCC$ is differentially simple.
\end{lemma}

\begin{proof}
  Let $\mathcal{E}$ be a class of graphs of locally bounded treewidth and $\psi(x,y)$ an interpretation formula such that $\CCC = I_{\psi}(\mathcal{E})$. From Gaifman's theorem applied to $\psi(x,y)$ it follows that there exist $d$ and $q$ such that the following holds for any $G \in \mathcal{E}$ and $H\in \CCC$ such that $H = I_{\psi}(G)$. If $u,v$ are two vertices of the same $d$-local $q$-type in $G \in \mathcal{E}$ and vertex $w$ is at distance more than $2d$ from both $u$ and $v$ in $G$, then $G\models \psi(u,w)$ iff  $G\models \psi(v,w)$, which in turn means that $uw \in E(H)$ iff $vw \in E(H)$. It follows that if $u,v$ are two vertices of $H$ such that in $G$ these vertices have the same $d$-local $q$-types, then every vertex in $D^H(u,v)$ has to be in the $2d$-neighbourhood of $u$ or $v$ in $G$.
  
  We define the colouring of any $H \in \CCC$ as follows. Let $G \in \mathcal{E}$ be such that $H = I_{\psi}(G)$. We colour every vertex $v$ of $H$ by its $t$-local $q$-type in $G$. By the above considerations for any two vertices $u,v \in V(H)$ of the same colour it has to hold that every vertex in $D^H(u,v)$ has to come from $N^G_{2t}(u) \cup N^G_{2t}(v)$. If we consider any two vertices $u',v' \in D^H(u,v)$ of the same colour, the same argumentation applies -- every vertex $w$ in $D^H(u',v')$ has to come from $N^G_{2t}(u') \cup N^G_{2t}(v')$ and thus has to be at distance at most $2d$ from $u$ or $v$ in $G$, which means $w \in N^G_{2d}(u) \cup N^G_{2d}(v)$.  It follows by an easy inductive argument that $DN^H_r(u,v)$ in $H$ is a subset of $N_{2dr}^G(u) \cup N_{2dr}^G(v)$ for any positive integer $r$. Since $\mathcal{E}$ is a class of graphs of locally bounded treewidth, the subgraph of $G$ induced by $N_{(r+1)2d}^G[u] \cup N_{(r+1)2d}^G[v]$ has treewidth bounded in terms of $(r+1)2d$ and  an easy argument shows that $H[DN^H_r[u,v]]$ is an induced subgraph of $I_\psi(G[N_{(r+1)2d}^G[u] \cup N_{(r+1)2d}^G[v]])$ which has bounded clique-width.
\end{proof}

Lemma~\ref{lem:lbdtw_diffsimple} implies that if we are able to efficiently compute the colourings from Definition~\ref{def:diff_simple}, then we obtain an efficient FO model checking algorithm for classes of graphs interpretable in graph classes of locally bounded treewidth by means of Theorem~\ref{thm:mcdiffsimple}. However, the existence of such colouring algorithm is unknown.

\section{Discussion and open problems}
\label{sec:open}

We have introduced the notions differential games and differential locality which can lead to efficient model checking algorithms and which seem to be more `interpretation friendly' than Gaifman's theorem. We believe that the ideas outlined in this paper can lead to improved understanding of the structure of graphs  interpretable in sparse graphs, and perhaps also lead to some insights into stable graphs (if Theorem~\ref{thm:interp_qw} gets strengthtened to stable graph classes).

\subsection{Complement-simple graph classes}
Regarding our application to the model checking problem for graph classes interpretable in classes of graphs with locally bounded treewidth, it has to be noted that there exists a simpler  approach based on colourings and on Gaifman's theorem and which avoids differential techniques altogether. 

\begin{definition}
\label{def:comp_simple}

We say that a class $\CCC$ of graphs is \emph{complement-simple} if for every $r$ there exists $r_m$ and graph class $\DDD_r$ with efficient FO model checking algorithm such that every $G \in \CCC$ has a $r_m$-colouring such that complementing edges between some pairs of colours results in a graph $G'$ in which for every $v \in V(G')$ it holds that $N^G_r[v] \in \DDD_r$.
\end{definition}

If $\CCC$ is a complement-simple graph class such that we can compute colourings from Definition~\ref{def:comp_simple} efficiently, then we can perform FO model checking on graphs from $\CCC$ efficiently. One can do this by noting that we can interpret $G$ in $G'$ and that we can solve the model checking problem on $G'$ efficiently by using Gaifman's theorem.

Coming back to graph classes interpretable in graph classes with locally bounded treewidth, using the colouring used in the proof of Lemma~\ref{lem:lbdtw_diffsimple} one can show that every such graph class $\CCC$ is complement-simple (and again one can use the same colouring for all values of $r$). In particular, complementing the edges in $G$ between some pairs of colours in this colouring leads to a graph $G'$ with locally bounded clique-width. We only briefly sketch the idea behind the proof here. Let $H\in \CCC$ and let $G$ be such that $H=I_{\psi}(G)$ and colour each vertex of $H$ by its $d$-local $q$-type, where $d$ and $q$ come from Gaifman's theorem applied to $\psi(x,y)$. We say that an edge $uv$ in $E(H)$ is \emph{long} if $dist_G(u,v)>2d$. Let $t_1:=tp_q^r(u)$ and $t_2:=tp_q(v)$. By Corollary~\ref{cor:gaifman} if there is a long edge in $H$ between any two vertices of types $t_1$ and $t_2$, then there exists an edge between all pairs of vertices of type $t_1$ and $t_2$ which are at distance more than $2d$ in $G$. In this case we say that types $t_1$ and $t_2$ \emph{induce long edges}. By complementing the edges between any pair of types (colours) in $H$ which induce long edges we remove all long edges in $H$ and obtain graph $H'$. It is easily shown that $H'$ is interpretable in $G$ (equipped with colours) by an interpretation which acts only locally and thus $H'$ has locally bounded clique-width. 

Similarly to the case of differentially simple graph classes, it is not clear whether one can compute the colourings from Definition~\ref{def:comp_simple} efficiently (fpt with respect to $r$) in the case of interpretations of graphs with locally bounded treewidth. In light of recent result of~\cite{rw_stable}, it is perhaps sensible to study graph classes which can be obtained form graph classes of locally bounded stable clique-width by a bounded number of complementations and then attempt to find an algorithm which `reverses' these complementations in the spirit of~\cite{BE_complements}.

Overall, the relationship between differentially simple and complement-simple graph classes is not clear and probably deserves further study. 

\subsection{Open problems}
We conclude with several open problems and possible directions for future research.
\begin{enumerate}
    \item \label{stable} Is it true that for any stable class $\CCC$ of graphs and any $q$ there exists $p$ such that every independent set in the graph of the relation $\cong_q^D$ has size at most $p$?
    \item Let $\CCC$ be a class of graphs interpretable in graph classes of locally bounded treewidth. Is there a polynomial algorithm which for every $G \in \CCC$ computes a colouring such that for every $r$ and every $u,v \in V(G)$ it holds that $G[DN_r[u,v]]$ has small clique-width (depending on $r$)? If not, is there an fpt algorithm which computes such colouring for every $r$?
    \item What is the relationship between differentially simple and complement-simple graph classes?
    \item Is it possible to use an approach based on differential games to give simpler/different algorithms for the FO model checking problem on graph classes of bounded expansion or nowhere dense graph classes than the algorithms presented in~\cite{dvorakkt10} and~\cite{gks14}? If yes, is it possible to use it to extend these results to interpretations of nowhere dense graph classes?
    \item More generally, if the answer to Question~\ref{stable} is yes, is it possible to use our methods to attack the FO model checking problem on stable graph classes?
    \item Is there a useful normal form for FO formulas (say, similar to Gaifman normal form) based on differential neighbourhoods and the formulas $\xi_r$?
    \item A recent result of~\cite{rw_stable} suggest the following questions. Is it true that a graph class $\CCC$ is interpretable in a class of graphs of locally bounded treewidth if and only if it is complement-simple, where each $\DDD_r$ from Definition~\ref{def:comp_simple} is a stable graph class of bounded clique-width? Is this true for differentially simple graph classes? The proof of Lemma~\ref{lem:lbdtw_diffsimple} suggests that one may consider these questions also with slightly adjusted (and weaker) definitions~\ref{def:comp_simple} and~\ref{def:diff_simple} in which we would require the existence of a single $m$ which works for every $r$.
    \item The approach to FO model checking outlined in Section~\ref{sec:mc} works for any class of graphs and not just on interpretations of nowhere dense graph classes (or possibly stable graph classes). Is it possible to use this approach to FO model checking to obtain efficient algorithms for FO model checking on non-stable graph classes? This would require a different way of computing relation $\sim_{q,t}$ from Corollary~\ref{cor:mc}.
\end{enumerate}

\bibliography{diffgames}
\bibliographystyle{plain}

\end{document}